\documentclass{llncs} 
\usepackage{graphicx} 
\usepackage{latexsym}
\usepackage{url}
\usepackage{xspace}
\usepackage{amssymb}
\usepackage{amsmath}

\newcommand{\UR}[1]{\mathit{UR}(#1)}
\newcommand{\ER}[1]{\mathit{ER}(#1)}

\newcommand{\LCL}{\theta}
\newcommand{\LCU}{\gamma}
\newcommand{\satequiv}{\equiv_{\mathit{sat}}}
\newcommand{\depqbf}{\textsf{DepQBF}\xspace}
\newcommand{\bloqqer}{\textsf{Bloqqer}\xspace}
\newcommand{\minisat}{\textsf{MiniSAT}\xspace}
\newcommand{\picosat}{\textsf{PicoSAT}\xspace}

\begin{document}
\author{Florian Lonsing \and Uwe Egly}
\institute{Vienna University of Technology \\ Institute of Information Systems \\ Knowledge-Based Systems Group
  \\ \url{http://www.kr.tuwien.ac.at/}}
\title{Incremental QBF Solving\thanks{Supported by the Austrian Science Fund (FWF) under
grant S11409-N23. We would like to thank Armin Biere and Paolo
Marin for helpful discussions. This article will appear in the proceedings of
the \emph{20th International Conference on Principles and Practice of Constraint
Programming, LNCS, Springer, 2014.}}}
\maketitle
\begin{abstract}
We consider the problem of incrementally solving a sequence of quantified
Boolean formulae (QBF). Incremental solving aims at using information learned
from one formula in the process of solving the next formulae in the
sequence. Based on a general overview of the problem and related challenges,
we present an approach to incremental QBF solving which is
application-independent and hence applicable to QBF encodings of arbitrary
problems. We implemented this approach in our incremental search-based QBF
solver \depqbf and report on implementation details. Experimental results
illustrate the potential benefits of incremental solving in QBF-based
workflows.
\end{abstract}


\section{Introduction}

The success of SAT technology in practical applications is largely driven by
\emph{incremental solving}.  SAT solvers based on conflict-driven clause
learning (CDCL)~\cite{DBLP:series/faia/SilvaLM09} gather information about a
formula in terms of learned clauses. When solving a sequence of closely related
formulae, it is beneficial to keep clauses learned from one formula in the
course of solving the next formulae in the sequence.

The logic of quantified Boolean formulae (QBF) extends propositional logic by
universal and existential quantification of variables. QBF potentially allows
for more succinct encodings of PSPACE-complete problems than SAT. 
Motivated by the success of incremental SAT solving, we consider the problem
of incrementally solving a sequence of syntactically related QBFs in prenex conjunctive
normal form (PCNF). 
 Building on search-based QBF solving with clause and cube learning
(QCDCL)~\cite{DBLP:journals/jar/CadoliSGG02,DBLP:journals/jair/GiunchigliaNT06,DBLP:conf/tableaux/Letz02,DBLP:conf/sat/LonsingEG13,DBLP:conf/cp/ZhangM02},
we present an approach to incremental QBF solving, which we
implemented in our solver \depqbf.\footnote{DepQBF is free software:
\url{http://lonsing.github.io/depqbf/}} 

Different from many incremental SAT and QBF~\cite{DBLP:conf/date/MarinMLB12}
solvers, \depqbf allows to add clauses to and delete clauses from the input
PCNF in a stack-based way by \emph{push} and \emph{pop} operations. A related
stack-based framework was implemented in the SAT solver
\picosat~\cite{DBLP:journals/jsat/Biere08}. A solver API with \emph{push} and
\emph{pop} increases the usability from the perspective of a user. Moreover, we present an 
optimization based on this stack-based framework which reduces the size of the learned clauses.

Incremental QBF solving was introduced for QBF-based
bounded model checking (BMC) of partial
designs~\cite{DBLP:conf/sat/MarinMB12,DBLP:conf/date/MarinMLB12}. This
approach, like ours, relies on selector variables and assumptions to support the deletion of clauses from the current
input
PCNF~\cite{DBLP:conf/sat/AudemardLS13,DBLP:journals/entcs/EenS03,DBLP:conf/sat/LagniezB13,DBLP:conf/sat/NadelR12}. The
quantifier prefixes of the incrementally solved PCNFs resulting from the BMC
encodings are modified only at the left or right end. In contrast to that, we
consider incremental solving of \emph{arbitrary} sequences of
PCNFs.  
For the soundness it is crucial to
determine which of the learned clauses and cubes can be kept across different
runs of an incremental QBF solver.  
We aim at a
general presentation of incremental QBF solving and illustrate problems
related to clause and cube learning. Our approach is
\emph{application-independent} and applicable to QBF encodings of
\emph{arbitrary} problems.

We report on experiments with constructed benchmarks. In addition to
 experiments with QBF-based conformant planning
using \depqbf~\cite{DBLP:journals/corr/EglyKLP14}, our results illustrate the potential benefits of incremental QBF solving 
 in application domains like
synthesis~\cite{DBLP:conf/vmcai/BloemKS14,DBLP:conf/sat/StaberB07}, formal
verification~\cite{jsat:benedetti08},
testing~\cite{DBLP:conf/date/HillebrechtKEWB13,DBLP:journals/tc/MangassarianVB10,DBLP:conf/iscas/SulflowFD10},
planning~\cite{DBLP:conf/ecai/CashmoreFG12}, and 
model enumeration~\cite{DBLP:conf/atva/BeckerELM12}, for example.


\section{Preliminaries} \label{sec_prelims}

We introduce terminology related to QBF and search-based QBF solving necessary to present a general
view on incremental solving.

For a propositional variable $x$, $l := x$ or $l := \neg x$ is a
\emph{literal}, where $v(l) = x$ denotes the variable of $l$.  A \emph{clause}
(\emph{cube}) is a disjunction (conjunction) of literals.  A \emph{constraint}
is a clause or a cube. The empty constraint $\emptyset$ does not contain any
literals. A clause (cube) $C$ is \emph{tautological} (\emph{contradictory}) if
$x \in C$ and $\neg x \in C$.

A propositional formula is in  \emph{conjunctive (disjunctive) normal form} if it consists of
a conjunction (disjunction) of clauses (cubes), called CNF (DNF). 
For simplicity, we regard
CNFs and DNFs as sets of clauses and cubes, respectively.

A quantified Boolean formula (QBF) $\psi := \hat{Q}.\,\phi$ is in  \emph{prenex CNF (PCNF)}
if it consists of a quantifier-free CNF $\phi$ and a  \emph{quantifier prefix}
$\hat{Q}$ with $\hat{Q} := Q_1B_1 \ldots Q_nB_n$ where $Q_i \in
\{\forall,\exists\}$ are  \emph{quantifiers} and $B_i$ are  \emph{blocks} (i.e.~sets) of
 variables such that $B_i \not = \emptyset$ and $B_i \cap B_j =
\emptyset$ for $i \not = j$, and $Q_i \not = Q_{i+1}$.

The blocks in the quantifier prefix are  \emph{linearly ordered} such that
$B_i < B_{j}$ if $i < j$. The linear ordering is extended to variables and literals:
$x_i < x_j$ if $x_i \in B_i$, $x_j \in B_j$ and $B_i < B_j$, and $l < l'$ if
$v(l) < v(l')$ for literals $l$ and $l'$. 

We consider only \emph{closed} PCNFs, where every variable which occurs in the
CNF is quantified in the prefix, and vice versa. 

A variable $x \in B_i$ is  \emph{universal}, written as $q(x) = \forall$, if $Q_i =
\forall$ and  \emph{existential}, written as $q(x) = \exists$, if $Q_i = \exists$. A
literal $l$ is universal if $q(v(l)) = \forall$ and existential if $q(v(l)) =
\exists$, written as $q(l) := \forall$ and $q(l) := \exists$, respectively.

An  \emph{assignment} is a mapping from variables to the truth values \emph{true}
 and \emph{false}. An assignment $A$ is represented as a set of
literals $A := \{l_1,\ldots,l_k\}$ such that, for $l_i \in A$, if $v(l_i)$ is
assigned to false (true) then $l_i = \neg v(l_i)$ ($l_i = v(l_i)$).

A \emph{PCNF $\psi$ under an assignment $A$} is denoted by $\psi[A]$ and is obtained
from $\psi$ as follows: 
for $l_i \in A$, if $l_i = v(l_i)$ ($l_i = \neg
v(l_i)$) then all occurrences of $v(l_i)$ in $\psi$ are replaced by the
syntactic truth constant $\top$ ($\bot$), respectively. All constants are
eliminated from $\psi[A]$ by the usual simplifications of Boolean algebra and
superfluous quantifiers and blocks are deleted from the quantifier prefix of
$\psi[A]$. Given a cube $C$ and a PCNF $\psi$, $\psi[C] := \psi[A]$ is the
formula obtained from $\psi$ under the assignment $A := \{l \mid l \in C\}$
defined by the literals in $C$.

The  \emph{semantics} of closed PCNFs is defined recursively. 
The QBF $\top$ is satisfiable and the QBF $\bot$ is
unsatisfiable. The QBF $\psi = \forall (B_1 \cup \{x\}) \ldots Q_nB_n.\,\phi$ is
satisfiable if $\psi[\neg x]$ and $\psi[x]$ are satisfiable. The QBF $\psi =
\exists (B_1 \cup \{x\}) \ldots Q_nB_n.\,\phi$ is satisfiable if $\psi[\neg x]$
or $\psi[x]$ are satisfiable. 

A PCNF $\psi$ is \emph{satisfied under an assignment} $A$ if $\psi[A] = \top$
and \emph{falsified} under $A$ if $\psi[A] = \bot$. Satisfied
and falsified clauses are defined analogously.

Given a constraint $C$, $L_Q(C) := \{l \in C \mid q(l) = Q\}$ for $Q \in
\{\forall,\exists\}$ denotes the set of universal and existential literals in
$C$. For a clause $C$, \emph{universal reduction} produces the
clause $\UR{C} := C \setminus \{l \mid l \in L_\forall(C) \textnormal{
and } \forall l' \in L_{\exists}(C): l' < l\}$.

\emph{Q-resolution} of clauses is a combination of resolution for propositional logic
and universal reduction~\cite{DBLP:journals/iandc/BuningKF95}. Given two
non-tautological clauses $C_1$ and $C_2$ and a pivot variable $p$ such that
$q(p) = \exists$ and $p \in C_1$ and $\neg p \in C_2$. Let $C' := (\UR{C_1}
\setminus \{p\}) \cup (\UR{C_2} \setminus \{\neg p\})$ be the \emph{tentative
Q-resolvent} of $C_1$ and $C_2$. If $C'$ is non-tautological then it is the
 \emph{Q-resolvent} of $C_1$ and $C_2$ and we write $C' = C_1 \otimes C_2$. Otherwise,
$C_1$ and $C_2$ do not have a Q-resolvent.

Given a PCNF $\psi := \hat{Q}.\,\phi$, a \emph{Q-resolution derivation} of a
clause $C$ from $\psi$ is the successive application of Q-resolution and
universal reduction to clauses in $\psi$ and previously derived clauses
resulting in $C$. We represent a derivation as a directed acyclic
graph (DAG) with edges (1) $C'' \rightarrow C'$ if $C' = \UR{C''}$ and (2)
$C_1 \rightarrow C'$ and $C_2 \rightarrow C'$ if $C' = C_1 \otimes C_2$. 
We write $\hat{Q}. \phi \vdash C$ if there is a derivation of a clause $C$
from $\psi$. Otherwise, we write $\hat{Q}. \phi \nvdash
C$. Q-resolution is a sound and refutationally-complete proof system for
QBFs~\cite{DBLP:journals/iandc/BuningKF95}. A \emph{Q-resolution proof} of an
unsatisfiable PCNF $\psi$ is a Q-resolution derivation of the empty clause.


\section{Search-Based QBF Solving} \label{sec_qcdcl_description}

We briefly describe search-based QBF solving with conflict-driven clause
learning and solution-driven cube learning
(QCDCL)~\cite{DBLP:journals/jar/CadoliSGG02,DBLP:journals/jair/GiunchigliaNT06,DBLP:conf/tableaux/Letz02,DBLP:conf/sat/LonsingEG13,DBLP:conf/cp/ZhangM02}
and related properties. 
In the context of \emph{incremental} QBF solving, clause and cube learning
requires a special treatment, which we address in
Section~\ref{sec_inc_satisfiability}.

Given a PCNF $\psi$, a QCDCL-based QBF solver successively assigns the
variables to generate an assignment $A$.  If $\psi$ is falsified under $A$,
i.e.~$\psi[A] = \bot$, then a new learned clause $C$ is derived by
Q-resolution and added to $\psi$. If $\psi$ is unsatisfiable, then finally the
empty clause will be derived by clause learning. If $\psi$ is satisfied under $A$,
i.e.~$\psi[A] = \top$, then a new learned \emph{cube} is constructed based on
the following \emph{model generation rule}, \emph{existential reduction} and
\emph{cube resolution}.

\begin{definition}[model generation
rule~\cite{DBLP:journals/jair/GiunchigliaNT06}] \label{def_model} Given a PCNF
$\psi := \hat{Q}. \phi$, an assignment $A$ such that $\psi[A] = \top$ is a 
\emph{model\footnote{We adopted this definition of models from~\cite{DBLP:conf/tableaux/Letz02}.} of $\psi$.} An \emph{initial
cube} $C = (\bigwedge_{l_i \in A} l_i)$ is a conjunction over the literals of
a model $A$. 
\end{definition}

\begin{definition}[\cite{DBLP:journals/jair/GiunchigliaNT06}]
Given a cube $C$, \emph{existential reduction} produces the reduced
cube $\ER{C} := C \setminus
\{l \mid l \in L_{\exists}(C) \textnormal{ and } \forall l' \in
L_{\forall}(C): l' < l\}$. 
\end{definition}

\begin{definition}[cube
resolution~\cite{DBLP:journals/jair/GiunchigliaNT06,DBLP:conf/cp/ZhangM02}]
Given two non-contradictory cubes $C_1$ and $C_2$, \emph{cube resolution} is
defined analogously to Q-resolution for clauses, except that existential
reduction is applied and the pivot variable must be universal. The cube
resolvent of $C_1$ and $C_2$ (if it exists) is denoted by $C := C_1 \otimes
C_2$.
\end{definition}

If $\psi$ is satisfiable, then finally the empty cube will be derived by cube
learning (Theorem~5 in~\cite{DBLP:journals/jair/GiunchigliaNT06}).  Whereas in
clause learning initially clauses of the input PCNF $\psi$ can be resolved, in
cube learning first initial cubes have to be generated by the model generation
rule, which can then be used to produce cube resolvents. 
Similar to Q-resolution derivations (DAGs) of clauses and Q-resolution proofs, we define \emph{cube resolution
derivations} of cubes and \emph{proofs of satisfiability}. 

\begin{figure}[t]
\hspace{0.12\textwidth} \emph{Clause derivation:} \hspace{0.28\textwidth} \emph{Cube derivation:}

\bigskip

\begin{minipage}{0.5\textwidth} 
\centering
\vspace{-0.3cm}
\includegraphics{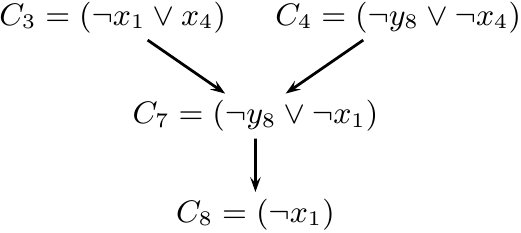}
\end{minipage}
\begin{minipage}{0.5\textwidth}
\centering
\includegraphics{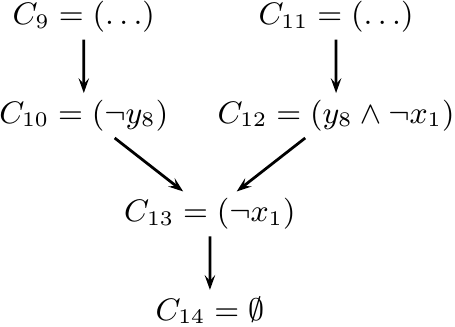}
\end{minipage}
\caption{Derivation DAGs of the clauses and cubes from
  Example~\ref{ex_running}. The literals in the initial cubes $C_9$ and
  $C_{11}$ have been omitted in the figure to save space.}
\label{fig_derivations}
\end{figure}

\begin{example} \label{ex_running}
Given the satisfiable PCNF $\psi
:= \exists x_1 \forall y_8 \exists x_5,x_2,x_6,x_4.\,\phi$, where $\phi :=
\bigwedge_{i := 1,\ldots, 6} C_i$ with $C_{1} := (y_8 \vee \neg x_5)$, $C_{2}
:= (x_2 \vee \neg x_6)$, $C_{3} := (\neg x_1 \vee x_4)$, $C_{4} := (\neg y_8
\vee \neg x_4)$, $C_{5} := (x_1 \vee x_6)$, and $C_{6} := (x_4 \vee x_5)$.

Figure~\ref{fig_derivations} shows the derivation of the clauses $C_7 := C_{3} \otimes
C_{4} = (\neg y_8 \vee \neg x_1)$ and $C_8 := \UR{C_7} = (\neg x_1)$ by
Q-resolution and universal reduction.

The assignment $A_1 := \{x_6, x_2, \neg y_8, \neg x_5,x_4\}$ is
a model of $\psi$ by Definition~\ref{def_model}. 
Hence $C_9 := (x_6 \wedge x_2 \wedge \neg y_8 \wedge \neg x_5 \wedge x_4)$ is an initial cube. Existential
reduction of $C_9$ produces the cube $C_{10} := \ER{C_9} = (\neg y_8)$. 
Similarly, $A_2 := \{y_8, \neg x_4, \neg x_1, x_5, x_6,
x_2\}$  is a model of $\psi$ and $C_{11} := (y_8 \wedge \neg x_4 \wedge
\neg x_1 \wedge x_5 \wedge x_6 \wedge x_2)$ is an initial cube.  Existential
reduction of $C_{11}$ produces the cube $C_{12} := \ER{C_{11}} = (y_8 \wedge \neg
x_1)$. The cube $C_{13} := (\neg x_1)$ is obtained by resolving $C_{10} = (\neg y_8)$ and $C_{12} = (y_8 \wedge
\neg x_1)$. Finally, existential
reduction of $C_{13}$ produces the empty cube $C_{14} := \ER{C_{13}} = \emptyset$, which proves that the PCNF $\psi$
is satisfiable.
\end{example}

A QCDCL-based solver implicitly constructs derivation DAGs in constraint
learning. However, typically only selected constraints of these derivations
are kept as learned constraints in an \emph{augmented
CNF}~\cite{DBLP:conf/cp/ZhangM02}.

\begin{definition} \label{def_aug_cnf}
Let $\psi := \hat{Q}.\,\phi$ be a PCNF. The \emph{augmented CNF (ACNF)} of $\psi$ has the form $\psi' :=
\hat{Q}.\,(\phi \wedge \LCL \vee \LCU)$, where $\hat{Q}$ is the
quantifier prefix, $\phi$ is the set of original clauses, $\LCL$
is a CNF containing the learned clauses, and $\LCU$ is a DNF
containing the learned cubes obtained by clause and cube learning in
QCDCL.
\end{definition} 

Given an ACNF $\psi'$ and an assignment $A$, the notation $\psi'[A]$ is
defined similarly to PCNFs. Analogously to clause derivations, we write
$\hat{Q}.\,\phi \vdash C$ if there is a derivation of a cube $C$ from the PCNF
$\hat{Q}. \phi$. During a run of a QCDCL-based solver the learned constraints can be derived from the current PCNF.

\begin{proposition} \label{prop_lcl_lcu_properties}
Let $\psi' := \hat{Q}.\,(\phi \wedge \LCL \vee \LCU)$ be the ACNF obtained by QCDCL from a PCNF $\psi := \hat{Q}.\,\phi$. It holds that
(1) $\forall C \in \LCL: \hat{Q}.\,\phi \vdash C$ and (2) $\forall C \in \LCU:
\hat{Q}.\,\phi \vdash C$.
\end{proposition}

Proposition~\ref{prop_lcl_lcu_properties} follows from the correctness of
constraint learning in \emph{non-incre\-mental} QCDCL. That is, we assume that
the PCNF $\psi$ is not modified over time. However, as we point out below, in
\emph{incremental} QCDCL the constraints learned previously might no longer be
derivable after the PCNF has been modified.

\begin{definition} \label{def_constraint_correct}
Given the ACNF $\psi' := \hat{Q}.\,(\phi \wedge \LCL \vee \LCU)$ of the PCNF
$\psi := \hat{Q}.\,\phi$, a clause $C \in
\LCL$ (cube $C \in \LCU$) is \emph{derivable} with respect to $\psi$ if $\psi
\vdash C$. Otherwise, if $\psi \nvdash C$, then $C$ is \emph{non-derivable}.
\end{definition}

Due to the correctness of model generation, existential/universal reduction,
and resolution, constraints which are derivable from the PCNF $\psi$ can be
added to the ACNF $\psi'$ of $\psi$, which results in a
satisfiability-equivalent ($\satequiv$) formula.

\begin{proposition}[\cite{DBLP:journals/jair/GiunchigliaNT06}] \label{prop_learning_satequiv} 
Let $\psi' := \hat{Q}.\,(\phi \wedge \LCL \vee \LCU)$ be the ACNF of the PCNF
$\psi := \hat{Q}.\,\phi$. Then (1)
$\hat{Q}. \phi \satequiv \hat{Q}. (\phi \wedge \LCL)$ and (2) $\hat{Q}. \phi
\satequiv \hat{Q}. (\phi \vee \LCU)$.
\end{proposition}


\section{Incremental Search-Based QBF Solving} \label{sec_inc_satisfiability}

We define \emph{incremental QBF solving} as the problem of solving a sequence
 of \mbox{PCNFs}  $\psi_0, \psi_1,\ldots, \psi_n$ using a QCDCL-based solver. Thereby,
the goal is to not discard all the learned constraints after
the PCNF $\psi_i$ has been solved. 
 Instead, to the largest extent possible we
want to re-use the constraints that were learned from $\psi_i$ in the process
of solving the next PCNF $\psi_{i+1}$. To this end, the ACNF $\psi_{i+1}' =
\hat{Q}_{i+1}.\,(\phi_{i+1} \wedge \LCL_{i+1} \vee \LCU_{i+1})$ of $\psi_{i+1}$
for $i > 0$, which is maintained by the solver, must be initialized with a set
$\LCL_{i+1}$ of learned clauses and a set $\LCU_{i+1}$ of learned cubes such
that $\LCL_{i+1} \subseteq \LCL_{i}$, $\LCU_{i+1} \subseteq \LCU_{i}$ and
 Proposition~\ref{prop_learning_satequiv} holds with respect to $\psi_{i+1}$. The sets $\LCL_{i}$ and
$\LCU_{i}$ contain the clauses and cubes that were learned from the previous
PCNF $\psi_i$ and potentially can be used to derive further 
constraints from $\psi_{i+1}$. If $\LCL_{i+1} \not = \emptyset$ and
$\LCU_{i+1} \not = \emptyset$ at the beginning, then the solver solves the PCNF
$\psi_{i+1}$ \emph{incrementally}. For the first PCNF $\psi_0$ in the sequence, the
solver starts with empty sets of learned constraints in the ACNF $\psi_{0}' =
\hat{Q}_{0}.\,(\phi_{0} \wedge \LCL_{0} \vee \LCU_{0})$.

Each PCNF $\psi_{i+1}$ for $0 \leq i < n$ in the sequence $\psi_0,\psi_1,
\ldots, \psi_n$ has the form $\psi_{i+1} = \hat{Q}_{i+1}.\,\phi_{i+1}$.  The
CNF part $\phi_{i+1}$ of $\psi_{i+1}$ results from $\phi_{i}$ of the previous
PCNF $\psi_{i} = \hat{Q}_{i}.\,\phi_{i}$ in the sequence by addition and
deletion of clauses.  We write $\phi_{i+1} = (\phi_{i} \setminus
\phi^{\mathit{del}}_{i+1} ) \cup \phi^{\mathit{add}}_{i+1}$, where
$\phi^{\mathit{del}}_{i+1}$ and $\phi^{\mathit{add}}_{i+1}$ are the sets of
deleted and added clauses. The quantifier prefix $\hat{Q}_{i+1}$ of
$\psi_{i+1}$ is obtained from $\hat{Q}_{i}$ of $\psi_{i}$ by deletion and
addition of variables and quantifiers, depending on the clauses in
$\phi^{\mathit{add}}_{i+1}$ and $\phi^{\mathit{del}}_{i+1}$. That is, we
assume that the PCNF $\psi_{i+1}$ is closed and that its prefix
$\hat{Q}_{i+1}$ does not contain superfluous quantifiers and variables.

When solving the PCNF $\psi_i$ using a QCDCL-based QBF solver, learned clauses
and cubes accumulate in the corresponding ACNF $\psi_{i}' =
\hat{Q}_{i}.\,(\phi_{i} \wedge \LCL_{i} \vee \LCU_{i})$. Assume that the
learned constraints are derivable with respect to $\psi_i$. 
 The PCNF $\psi_i$ is modified
to obtain the next PCNF $\psi_{i+1}$ to be solved. The learned
constraints in $\LCL_i$ and $\LCU_i$ might become non-derivable with respect to
$\psi_{i+1}$ in the sense of
Definition~\ref{def_constraint_correct}. Consequently,
Proposition~\ref{prop_learning_satequiv} might no longer hold for the ACNF
$\psi_{i+1}' = \hat{Q}_{i+1}.\,(\phi_{i+1} \wedge \LCL_{i+1} \vee \LCU_{i+1})$
of the new PCNF $\psi_{i+1}$ if previously learned constraints from $\LCL_i$
and $\LCU_{i}$ appear in $\LCL_{i+1}$ and $\LCU_{i+1}$. In this case, the solver
might produce a wrong result when solving $\psi_{i+1}$. 

\subsection{Clause Learning} \label{sec_inc_cdcl}

Assume that the PCNF $\psi_i = \hat{Q}_{i}.\,\phi_{i}$ has been solved and
learned constraints have been collected in the ACNF $\psi_i' =
\hat{Q}_{i}.\,(\phi_{i} \wedge \LCL_{i} \vee \LCU_{i})$.  The clauses in
$\phi^{\mathit{del}}_{i+1}$ are deleted from $\phi_i$ to obtain the CNF part
$\phi_{i+1} = (\phi_i \setminus \phi^{\mathit{del}}_{i+1}) \cup
\phi^{\mathit{add}}_{i+1}$ of the next PCNF $\psi_{i+1} = \hat{Q}_{i+1}.\,\phi_{i+1}$. If the
 derivation of a learned clause $C \in \LCL_i$ depends on deleted
clauses in $\phi^{\mathit{del}}_{i+1}$, then we might have that
$\psi_i \vdash C$ but $\psi_{i+1} \nvdash C$. In this case, 
 $C$ is non-derivable with respect to the next PCNF $\psi_{i+1}$. Hence 
$C$ must be discarded before solving $\psi_{i+1}$ starts so
that $C \not \in \LCL_{i+1}$ in the initial ACNF $\psi_{i+1}' =
\hat{Q}_{i+1}.\,(\phi_{i+1} \wedge \LCL_{i+1} \vee
\LCU_{i+1})$. Otherwise, if $C \in \LCL_{i+1}$ then the solver might
construct a bogus Q-resolution proof for the PCNF $\psi_{i+1}$ and, if
$\psi_{i+1}$ is satisfiable, erroneously conclude that $\psi_{i+1}$ is
unsatisfiable.

\begin{example} \label{ex_clause_delete_incorrect_clause}
Consider the PCNF $\psi$ from Example~\ref{ex_running}. The derivation of the clause $C_8 = (\neg x_1)$ shown in
Fig.~\ref{fig_derivations} depends on
the clause $C_4 = (\neg y_8 \vee \neg x_4)$. We have that $\psi \vdash C_8$. Let $\psi_1$ be the PCNF obtained
from $\psi$ by deleting $C_4$. Then $\psi_1 \nvdash C_8$ because $C_3 = (\neg
x_1 \vee x_4)$ is
the only clause which contains the literal $\neg x_1$. Hence a possible
derivation of the clause $C_8 = (\neg x_1)$ must use $C_3$. However, no such
derivation exists in $\psi_1$. There is no clause $C'$ containing a literal
$\neg x_4$ which can be resolved with $C_3$ to produce $C_8 = (\neg
x_1)$ after a sequence of resolution steps. 
\end{example}

Consider the PCNF $\psi_{i+1} = \hat{Q}_{i+1}.\,\phi_{i+1}$ with $\phi_{i+1} =
\phi_i \cup \phi^{\mathit{add}}_{i+1}$ which is obtained from
$\hat{Q}_{i}.\,\phi_{i}$ by \emph{only adding} the clauses
$\phi^{\mathit{add}}_{i+1}$, but not deleting any clauses. 
Assuming that $\hat{Q}_i. \phi_i
\vdash C$ for all $C \in \LCL_i$ in the ACNF $\psi_i' =
\hat{Q}_{i}.\,(\phi_{i} \wedge \LCL_{i} \vee \LCU_{i})$, also
$\hat{Q}_{i+1}.\,(\phi_i \cup \phi^{\mathit{add}}_{i+1}) \vdash C$. Hence all the
learned clauses in $\LCL_i$ are derivable with respect to the next PCNF
$\psi_{i+1}$ and can be added to the ACNF $\psi_{i+1}'$.

\subsection{Cube Learning} \label{sec_inc_sdcl}

Like above, let $\psi_i' = \hat{Q}_{i}.\,(\phi_{i} \wedge \LCL_{i} \vee
\LCU_{i})$ be the ACNF of the previously solved PCNF $\psi_i =
\hat{Q}_{i}.\,\phi_{i}$.  Dual to clause deletions, the addition of clauses
to $\phi_i$ can make learned cubes in $\LCU_i$ non-derivable with respect to the
next PCNF $\psi_{i+1} = \hat{Q}_{i+1}.\,\phi_{i+1}$ to be solved. The clauses
in $\phi^{\mathit{add}}_{i+1}$ are added to $\phi_i$ to obtain the CNF part
$\phi_{i+1} = (\phi_i \setminus \phi^{\mathit{del}}_{i+1}) \cup
\phi^{\mathit{add}}_{i+1}$ of $\psi_{i+1}$.  An initial
cube $C \in \LCU_i$ has been obtained from a model $A$ of the previous PCNF
$\psi_i$, i.e.~$\psi_i[A] = \top$. We might have that
$\psi_{i+1}[A] \not = \top$ with respect to the next PCNF $\psi_{i+1}$ because
of an added clause $C' \in \phi^{\mathit{add}}_{i+1}$ (and hence also $C' \in
\phi_{i+1}$) such that $C'[A] \not = \top$. Therefore, $A$ is not a model of
$\psi_{i+1}$ and the initial cube $C$ is non-derivable with
respect to $\psi_{i+1}$, i.e.~$\hat{Q}_i. \phi_i \vdash C$ but $\hat{Q}_{i+1}. \phi_{i+1} \nvdash C$.  Hence
$C$ and every cube whose derivation depends on $C$ must be discarded to prevent the solver from
generating a bogus cube resolution proof for $\psi_{i+1}$.  If $\psi_{i+1}$ is
unsatisfiable, then the solver might erroneously conclude that $\psi_{i+1}$ is
satisfiable. That is, Proposition~\ref{prop_learning_satequiv} might not hold
with respect to non-derivable cubes and the ACNF $\psi_{i+1}'$ of
$\psi_{i+1}$.

\begin{example} \label{ex_clause_add_incorrect_cube}
Consider the PCNF $\psi$ from Example~\ref{ex_running}.  The
derivation of the cube $C_{10} = (\neg y_8)$ shown in
Fig.~\ref{fig_derivations} depends on the initial cube $C_9 = (x_6
\wedge x_2 \wedge \neg y_8 \wedge \neg x_5 \wedge x_4)$, which has
been generated from the model $A_1 = \{x_6, x_2, \neg y_8, \neg
x_5,x_4\}$. The cube $C_9$ is derivable with respect to $\psi$ since
$\psi[A_1] = \top$, and hence $\psi \vdash C_9$. The cube $C_{10}$ is also
derivable since $C_{10} = \ER{C_9}$.  
Assume that the clause $C_0 := (\neg x_2 \vee \neg x_4)$ is added to
$\psi$ resulting in the unsatisfiable PCNF $\psi_2$. Now $C_9$ is non-derivable with respect
to $\psi_2$ since $C_0[A_1] = \bot$. Further, $\psi_2 \nvdash
C_{10}$.
\end{example}

Consider the PCNF $\psi_{i+1} = \hat{Q}_{i+1}.\,\phi_{i+1}$ with $\phi_{i+1} =
\phi_i \setminus \phi^{\mathit{del}}_{i+1}$ which is obtained from
$\hat{Q}_{i}.\,\phi_{i}$ by \emph{only deleting} the clauses
$\phi^{\mathit{del}}_{i+1}$, but not adding any clauses. If after the clause deletions
 some variable $x$ does not occur
anymore in the resulting PCNF $\psi_{i+1}$, then $x$ is removed from the
quantifier prefix of $\psi_{i+1}$ and from every cube $C  \in \LCU_i$ which was learned
when solving the previous PCNF $\psi_{i}$.
Proposition~\ref{prop_learning_satequiv} holds for the cleaned up cubes $C' =
C \setminus \{l \mid v(l) = x\}$ for all $C \in \LCU_i$ with respect to
$\psi_{i+1}$ and hence $C'$ can be added to the ACNF $\psi_{i+1}'$.

\begin{proposition} \label{prop_cleaned_up_cubes_sound}
Let $\psi_{i}' := \hat{Q}_{i}.\,(\phi_{i} \wedge \LCL_{i} \vee \LCU_{i})$ be
the ACNF of the PCNF $\psi_{i} := \hat{Q}_{i}.\,\phi_{i}$. Let $\psi_{i+1} :=
\hat{Q}_{i+1}.\,\phi_{i+1}$ be the PCNF resulting from $\psi_{i}$ with $\phi_{i+1} = (\phi_{i} \setminus
\phi^{\mathit{del}}_{i+1})$, where the
variables $V^{\mathit{del}}_{i+1}$ no longer occur in $\phi_{i+1}$ and are
removed from $\hat{Q}_{i}$ to obtain $\hat{Q}_{i+1}$. Given a cube $C
\in \LCU_{i}$, let $C' := C \setminus \{l \mid v(l) \in
V^{\mathit{del}}_{i+1}\}$. Proposition~\ref{prop_learning_satequiv} holds for
$C'$ with respect to $\hat{Q}_{i+1}.\,\phi_{i+1}$: $\hat{Q}_{i+1}.\,\phi_{i+1} \satequiv
\hat{Q}_{i+1}.\,(\phi_{i+1} \vee C')$.
\end{proposition}
\begin{proof}[Sketch]
By induction on the structure of the derivations of cubes in $\LCU_i$.

Let $C \in \LCU_i$ be an initial cube due to the assignment $A$
with $\psi_i[A] = \top$. For $A' := A \setminus \{l \mid v(l) \in
V^{\mathit{del}}_{i+1}\}$, we have $\psi_{i+1}[A'] = \top$ since 
all the clauses containing the variables in $V^{\mathit{del}}_{i+1}$ were
deleted from $\psi_i$ to obtain $\psi_{i+1}$. Then the claim holds for the
initial cube $C' = C \setminus \{l \mid v(l) \in V^{\mathit{del}}_{i+1}\} = (\bigwedge_{l_i
  \in A'} l_i)$ since $\psi_{i+1} \vdash C'$.

Let $C \in \LCU_i$ be obtained from $C_1 \in \LCU_i$ by existential reduction
such that $C = \ER{C_1}$. Assuming that the claim holds for $C_1' = C_1
\setminus \{l \mid v(l) \in V^{\mathit{del}}_{i+1}\}$, it also
holds for $C' = C \setminus \{l \mid v(l) \in V^{\mathit{del}}_{i+1}\} = \ER{C_1'}$ since
existential reduction removes existential literals which are
maximal with respect to the prefix ordering.

Let $C \in \LCU_i$ be obtained from $C_1,C_2 \in \LCU_i$ by
resolution on variable $x$ with $x \in C_1$, $\neg x \in C_2$. Assume that the claim holds for $C_1' = C_1
\setminus \{l \mid v(l) \in V^{\mathit{del}}_{i+1}\}$ and $C_2' = C_2
\setminus \{l \mid v(l) \in V^{\mathit{del}}_{i+1}\}$, i.e.~$\hat{Q}_{i+1}.\,\phi_{i+1} \satequiv
\hat{Q}_{i+1}.\,(\phi_{i+1} \vee C_1')$ and $\hat{Q}_{i+1}.\,\phi_{i+1} \satequiv
\hat{Q}_{i+1}.\,(\phi_{i+1} \vee C_2')$. 
If $x \not \in V^{\mathit{del}}_{i+1}$ then the claim also holds for
$C' = C \setminus \{l \mid v(l) \in V^{\mathit{del}}_{i+1}\} = C_1'
\otimes C_2'$ with $x \in C_1'$, $\neg x \in C_2'$ due to the
correctness of resolution (Proposition~\ref{prop_learning_satequiv}). 
If $x \in V^{\mathit{del}}_{i+1}$ then the claim also holds for $C' =
C \setminus \{l \mid v(l) \in V^{\mathit{del}}_{i+1}\} = (C_1' \wedge
C_2')$ since $\{y, \neg y\} \not \subseteq (C_1' \cup C_2')$ for all
variables $y$, which can be proved by reasoning with tree-like models
of QBFs~\cite{DBLP:conf/cp/SamulowitzDB06}.
\qed
\end{proof}

If a variable $x$ no longer occurs in the formula, then cubes where $x$ has
been removed might become non-derivable. However, due to
Propositions~\ref{prop_learning_satequiv} and~\ref{prop_cleaned_up_cubes_sound}
it is sound to keep all the cleaned up cubes (resolution is not
inferentially-complete). Moreover, due to the correctness of resolution and
existential reduction, Proposition~\ref{prop_learning_satequiv} also holds for
new cubes derived from the cleaned up cubes.

In practice, the goal is to keep as many 
 learned constraints as possible because they prune the search space and
can be used to derive further constraints. Therefore, subsets
$\LCL_{i+1} \subseteq \LCL_{i}$ and $\LCU_{i+1} \subseteq \LCU_{i}$ of the
learned clauses $\LCL_{i}$ and cubes $\LCU_{i}$ must be selected so that
Proposition~\ref{prop_learning_satequiv} holds with respect to the initial
ACNF $\psi_{i+1}' = \hat{Q}_{i+1}.\,(\phi_{i+1} \wedge \LCL_{i+1} \vee
\LCU_{i+1})$ of the PCNF $\psi_{i+1}$ to be solved next.


\section{Implementing an Incremental QBF Solver}

We describe the implementation of our incremental
QCDCL-based solver \depqbf. Our approach is general and
fits any QCDCL-based solver. For incremental solving we
do not apply a sophisticated analysis of variable dependencies by dependency
schemes in \depqbf~\cite{DBLP:conf/sat/LonsingB10}. Instead, as many other QBF
solvers, we use the linear ordering given by the quantifier prefix. 
We implemented a stack-based representation of the CNF part of
PCNFs based on selector variables and assumptions. Assumptions were also used
for incremental QBF-based BMC of partial
designs~\cite{DBLP:conf/date/MarinMLB12} and are common in incremental SAT
solving~\cite{DBLP:conf/sat/AudemardLS13,DBLP:journals/entcs/EenS03,DBLP:conf/sat/LagniezB13,DBLP:conf/sat/NadelR12}.

We address the problem of checking which learned constraints can be kept across
different solver runs after the current PCNF has been modified. To this end, we
present approaches to check if a constraint learned from the previous
PCNF is still derivable from the next one, which makes sure that
Proposition~\ref{prop_learning_satequiv} holds. Similar to incremental SAT solving, 
selector variables are used to handle the learned clauses. Regarding
learned cubes, selector variables can also be used (although in a way
asymmetric to clauses), in addition to an
alternative approach relying on full derivation DAGs, which have to be kept in
memory. Learned cubes might become non-derivable by the
deletion of clauses and superfluous variables, but still can be kept due to
Proposition~\ref{prop_cleaned_up_cubes_sound}. We implemented a simple approach
which, after clauses have been added to the formula, allows to keep only
initial cubes but not cubes obtained by resolution or existential reduction.


\subsection{QBF Solving under Assumptions} \label{sec_qbf_assumptions}

Let $\psi := Q_1B_1Q_2B_2 \ldots Q_nB_n.\,\phi$ be a PCNF. We define a set $A
:= \{l_1,\ldots,l_k\}$ of \emph{assumptions}  as an assignment such that
$v(l_i) \in B_1$ for all literals $l_i \in A$. The variables assigned
by $A$ are from the first block $B_1$ of $\psi$.  Solving the PCNF $\psi$
\emph{under the set $A$ of assumptions} amounts to solving the PCNF
$\psi[A]$. The definition of assumptions can be applied recursively to the
PCNF $\psi[A]$. If $A$ assigns all the variables in $B_1$, then variables from
$B_2$ can be assigned as assumptions with respect to $\psi[A]$, since $B_2$ is the
first block in the quantifier prefix of $\psi[A]$.

We implemented the handling of assumptions according to the
\emph{literal-based single instance (LS)} approach (in the terminology
of~\cite{DBLP:conf/sat/NadelR12}). Thereby, the assumptions in $A$ are treated
in a special way so that the variables in $A$ are never selected as pivots
in the resolution derivation of a learned constraint according to QCDCL-based
learning.  Similar to SAT-solving under assumptions, LS allows to keep all the
constraints that were learned from the PCNF $\psi[A]$ under a set $A$ of
assumptions when later solving $\psi[A']$ under a different set $A'$ of
assumptions.


\subsection{Stack-Based CNF Representation} \label{sec_cnf_stack}

In \depqbf, the CNF part $\phi$ of an ACNF $\psi_i' =
\hat{Q}_i.\,(\phi_i \wedge \LCL_i \vee \LCU_i)$ to be solved is
represented as a stack of clauses. The clauses on the stack are
grouped into \emph{frames}. The solver API provides functions to push
new frames onto the stack, pop present frames from the stack, and to
add new clauses to the current topmost frame. Each \emph{push}
operation opens a new topmost frame $f_j$.  
New clauses are always added to the topmost
frame $f_j$. Each new frame $f_j$ opened by a \emph{push} operation is
associated with a fresh \emph{frame selector variable} $s_j$. Frame
selector variables are existentially quantified and put into a
separate, leftmost quantifier block $B_0$  
i.e.~the current ACNF $\psi_i'$ has the form $\psi_i' =
\exists B_0\hat{Q}_i.\,(\phi_i \wedge \LCL_i \vee \LCU_i)$. Before a
new clause $C$ is added to frame $f_j$, the frame selector variable $s_j$ of $f_j$ is
inserted into $C$ so that in fact the clause $C' = C \cup \{s_j\}$ is added
to $f_j$. If all the selector variables are assigned to \emph{false} then
under that assignment every clause $C' = C \cup \{s_j\}$ is syntactically equivalent to $C$. 

The purpose of the frame selector variables is to \emph{enable} or \emph{disable} the
clauses in the CNF part $\phi_i$ with respect to the \emph{push} and
\emph{pop} operations applied to the clause stack. If the selector variable
$s_j$ of a frame $f_j$ is assigned to \emph{true} then all the clauses of
$f_j$ are satisfied under that assignment. In this case, these
satisfied clauses are considered disabled because they can not be used to
derive new learned clauses in QCDCL. Otherwise, the assignment of
\emph{false} to $s_j$ does not satisfy any clauses in $f_j$. Therefore these
clauses are considered enabled.

Before the solving process starts, the clauses of frames popped from
the stack are disabled and the clauses of frames still on the stack are
enabled by assigning the selector variables to \emph{true} and \emph{false},
respectively.  The selector variables are assigned as 
assumptions. This is possible because these variables are in the leftmost
quantifier block $B_0$ of the ACNF $\psi_i' = \exists B_0\hat{Q}_i.\,(\phi_i \wedge
\LCL_i \vee \LCU_i)$ to be solved.

The idea of enabling and disabling clauses by selector variables and
assumptions originates from incremental SAT solving~\cite{DBLP:journals/entcs/EenS03}. This approach was also applied
to bounded model checking of partial designs by incremental QBF
solving~\cite{DBLP:conf/date/MarinMLB12}. In \depqbf, we implemented the
\emph{push} and \emph{pop} operations related to the clause stack by selector
variables  similarly to the SAT solver
\picosat~\cite{DBLP:journals/jsat/Biere08}.

  In the implementation of \depqbf, frame selector variables are maintained
entirely by the solver.  Depending on the \emph{push} and \emph{pop} operations, selector variables 
 are automatically inserted into added clauses and assigned  as assumptions. This
approach saves the user the burden of inserting selector variables manually
into the QBF encoding of a problem and assigning them as assumptions via the solver API. Manual insertion is typically applied in
incremental SAT solving based on assumptions as pioneered by
\minisat~\cite{DBLP:conf/sat/EenS03,DBLP:journals/entcs/EenS03}. We argue that
the usability of an incremental QBF solver is improved considerably if the
selector variables are maintained by the solver. For example, from the
perspective of the user, the QBF encoding contains only variables relevant to 
the encoded problem. 

In the following, we consider the problem of maintaining the sets of learned
constraints across different solver runs. As pointed out in
Section~\ref{sec_inc_satisfiability}, Proposition~\ref{prop_learning_satequiv}
still holds for learned clauses (cubes) after the addition (deletion) of
clauses to (from) the PCNF. Therefore, we present the maintenance of learned
constraints separately for clause additions and deletions.


\subsection{Handling Clause Deletions} \label{sec_clause_del}

A clause $C \in \LCL_i$ in the current
ACNF $\psi_i' = \hat{Q}_i.\,(\phi_i \wedge \LCL_i \vee \LCU_i)$ might become non-derivable if its derivation depends
on clauses in $\phi^{\mathit{del}}_{i+1}$ which are deleted to
obtain the CNF part $\phi_{i+1} = (\phi_{i} \setminus
\phi^{\mathit{del}}_{i+1} ) \cup \phi^{\mathit{add}}_{i+1}$ of the
next PCNF $\psi_{i+1}$.

In \depqbf, learned clauses in $\LCL_i$ are deleted as follows. As pointed out
in the previous section, clauses of popped off frames are disabled by
assigning the respective frame selector variables to \emph{true}.  Since the
formula contains only positive literals of selector
variables, these variables cannot be chosen as pivots in derivations. 
Therefore, learned clauses whose derivations depend on
disabled clauses of a popped off frame $f_j$ contain the selector variable
$s_j$ of $f_j$.  Hence these learned clauses are also disabled by the
assignment of $s_j$.  This
 approach to handling learned clauses is also
applied in incremental SAT solving~\cite{DBLP:journals/entcs/EenS03}. 

The disabled clauses are physically deleted in a garbage collection phase if
their number exceeds a certain threshold. Variables which no longer occur in
the CNF part of the current PCNF are removed from the quantifier prefix and,
by Proposition~\ref{prop_cleaned_up_cubes_sound}, from learned cubes in
$\LCU_i$ to produce cleaned up cubes. We initialize the set $\LCU_{i+1}$ of
learned cubes in the ACNF $\psi_{i+1}' = \hat{Q}_{i+1}.\,(\phi_{i+1} \wedge \LCL_{i+1} \vee
\LCU_{i+1})$ of the next PCNF $\psi_{i+1}$ to be solved to contain the cleaned
up cubes.

The deletion of learned clauses based on selector variables is not optimal in
the sense of Definition~\ref{def_constraint_correct}. There might be another
derivation of a disabled learned clause $C$ which does not depend on the
deleted clauses $\phi^{\mathit{del}}_{i+1}$. This observation also applies to
the use of selector variables in incremental SAT solving.
 
As illustrated in the context of incremental SAT solving, the size of learned
clauses might increase considerably due to the additional selector
variables~\cite{DBLP:conf/sat/AudemardLS13,DBLP:conf/sat/LagniezB13}. In the
stack-based CNF representation of \depqbf, the clauses
associated to a frame $f_j$ all contain the selector variable
$s_j$ of $f_j$. Therefore, the maximum number of selector variables in a new
clause learned from the current PCNF $\psi_i$ is bounded by the number of
currently enabled frames. The sequence of \emph{push} operations introduces a
linear ordering $f_0 < f_1 < \ldots < f_k$ on the enabled frames $f_i$ and
their clauses in the CNF with respect to the point of time where that frames
and clauses have been added. In \depqbf, we implemented the following
optimization based on this temporal ordering. Let $C$ 
and $C'$ be clauses which are resolved in the
course of clause learning. Assume that $s_i \in C$ and $s_j \in C'$ are the
only selector variables of currently enabled frames $f_i$ and $f_j$ in $C$ and
$C'$.  Instead of computing the usual Q-resolvent $C'' := C \otimes C'$, we
compute $C'' := (C \otimes C') \setminus \{l \mid l = s_i \textnormal{ if
} f_i < f_j \textnormal{ and } l = s_j \textnormal{ otherwise}\}$. That is,
the selector variable of the frame which is smaller in the temporal ordering is
discarded from the resolvent. If $f_i < f_j$ then the clauses in $f_i$ were
pushed onto the clause stack before the clauses in $f_j$. The frame $f_j$ will
be popped off the stack before $f_i$. Therefore, in order to properly disable
the learned clause $C''$ after \emph{pop} operations, it is sufficient to keep
the selector variable $s_j$ of the frame $f_j$ in $C''$. With this
optimization, \emph{every} learned clause contains \emph{exactly one} selector
variable. In the SAT solver \picosat, an optimization which has similar
effects is implemented.


\subsection{Handling Clause Additions} \label{sec_clause_add}

Assume that the PCNF $\psi_i := \hat{Q}_i.\,\phi_i$ has been
solved and that all learned constraints in  the ACNF $\psi_i' = \hat{Q}_i.\,(\phi_i \wedge \LCL_i \vee
\LCU_i)$ of $\psi_i$ are derivable with respect to
$\psi_{i}$. The set $\phi^{\mathit{add}}_{i+1}$ of clauses is
added to $\phi_i$ to obtain the CNF part $\phi_{i+1} = (\phi_{i} \setminus
\phi^{\mathit{del}}_{i+1} ) \cup \phi^{\mathit{add}}_{i+1}$ of the next PCNF
$\psi_{i+1} = \hat{Q}_{i+1}.\,\phi_{i+1}$. For learned clauses,
 we can set $\LCL_{i+1} :=
\LCL_{i}$ in the ACNF $\psi_{i+1}' = \hat{Q}_{i+1}.\,(\phi_{i+1} \wedge
\LCL_{i+1} \vee \LCU_{i+1})$ of $\psi_{i+1}$. The following example
illustrates  the effects of adding $\phi^{\mathit{add}}_{i+1}$ on the cubes.

\begin{example} \label{ex_incorrect_cubes_clean_dag}
Consider the cube derivation shown in Fig.~\ref{fig_derivations}. As illustrated in
Example~\ref{ex_clause_add_incorrect_cube}, the cubes $C_9 = (x_6 \wedge x_2
\wedge \neg y_8 \wedge \neg x_5 \wedge x_4)$ and $C_{10} = (\neg y_8)$ are
non-derivable with respect to the PCNF $\psi_2$ obtained from $\psi$ by adding the
clause $C_0 := (\neg x_2 \vee \neg x_4)$. The initial cube $C_{11} := (y_8
\wedge \neg x_4 \wedge \neg x_1 \wedge x_5 \wedge x_6 \wedge x_2)$ still is
derivable because the underlying model $A_2 := \{y_8, \neg x_4, \neg x_1, x_5, x_6, x_2\}$
of $\psi$ is also a model of
$\psi_2$. Therefore, when solving $\psi_2$ we can keep the derivable cubes $C_{11}$ and
$C_{12} = \ER{C_{11}}$.  
The non-derivable cubes $C_9$ and $C_{10}$ must be discarded. Otherwise,
QCDCL might produce the cube resolution proof shown in Fig.~\ref{fig_derivations}
when solving the \emph{unsatisfiable} PCNF $\psi_2$, which is incorrect.
\end{example}

We sketch an approach to identify the cubes in a cube derivation DAG $G$ which are
non-derivable with respect to the next PCNF $\psi_{i+1} =
\hat{Q}_{i+1}.\,\phi_{i+1}$. 
Starting at the initial cubes, $G$ is traversed in a topological order. An
initial cube $C$ is marked as derivable if $\psi_{i+1}[C] = \top$, otherwise if
$\psi_{i+1}[C] \not = \top$ then $C$ is marked as non-derivable. This test can be
carried out syntactically by checking whether every clause of $\psi_{i+1}$ is
satisfied under the assignment given by $C$. A cube $C$
obtained by existential reduction or cube resolution is marked as derivable if
all its predecessors in $G$ are marked as derivable. Otherwise, $C$ is marked as
non-derivable. Finally, all cubes in $G$ marked as non-derivable are
deleted.

The above procedure allows to find a subset $\LCU_{i+1} \subseteq
\LCU_{i}$ of the set $\LCU_{i}$ of cubes in the solved ACNF $\psi_i' =
\hat{Q}_i.\,(\phi_i \wedge \LCL_i \vee \LCU_i)$ so that
 all cubes in $\LCU_{i+1}$ are derivable and
Proposition~\ref{prop_learning_satequiv} holds for the next ACNF
$\psi_{i+1}' = \hat{Q}_{i+1}.\,(\phi_{i+1} \wedge \LCL_{i+1} \vee
\LCU_{i+1})$. However, this procedure is not optimal because it might mark a
cube $C \in G$ as non-derivable with respect to the next PCNF $\psi_{i+1}$
although $\psi_{i+1} \vdash C$.

\begin{example} \label{ex_cube_check_overapprox} Given the satisfiable PCNF
$\psi := \exists x_1 \forall y_8 \exists x_5,x_2,x_6,x_4.\,\phi$, where $\phi
:= \bigwedge_{i := 1,\ldots, 5} C_i$ with the clauses $C_i$ from
Example~\ref{ex_running} where $C_{1} := (y_8 \vee \neg x_5)$, $C_{2} := (x_2
\vee \neg x_6)$, $C_{3} := (\neg x_1 \vee x_4)$, $C_{4} := (\neg y_8 \vee \neg
x_4)$, $C_{5} := (x_1 \vee x_6)$.  Consider the model $A_3 := \{\neg x_1, y_8,
\neg x_5, x_2, x_6, \neg x_4\}$ of $\psi$ and the initial cube $C_{15} :=
(\neg x_1 \wedge y_8 \wedge \neg x_5 \wedge x_2 \wedge x_6 \wedge \neg x_4)$ generated
from $A_3$. Existential reduction of $C_{15}$ produces the cube $C_{16} :=
\ER{C_{15}} = (\neg x_1 \wedge y_8)$.  Assume that the clause $C_0 := (x_4 \vee
x_5)$ is added to $\psi$ to obtain the PCNF $\psi_3$. The initial cube
$C_{15}$ is non-derivable with respect to $\psi_3$ since $C_{0}[A_3] \not =
\top$. However, for the cube $C_{16}$ derived from $C_{15}$ it holds that
$\psi_3 \vdash C_{16}$. The assignment $A_4 := \{\neg x_1, y_8, x_5, x_2, x_6,
\neg x_4\}$ is a model
of $\psi_3$. Let $C_{17} := (\neg x_1 \wedge y_8 \wedge x_5 \wedge x_2 \wedge x_6 \wedge
\neg x_4)$ be the initial cube generated from $A_4$. Then $C_{16} =
\ER{C_{17}}$ is derivable with respect to $\psi_3$.
\end{example}

In practice, QCDCL-based solvers typically store only the learned cubes, which
might be a small part of the derivation DAG $G$, and no edges. Therefore, checking
the cubes in a traversal of $G$ is not feasible. Even if the full DAG $G$ is
available, the checking procedure is not optimal as pointed out in
Example~\ref{ex_cube_check_overapprox}. Furthermore, it cannot be
used to check cubes which have become non-derivable after cleaning up by
Proposition~\ref{prop_cleaned_up_cubes_sound}. Hence, it is desirable to have an approach to checking 
the derivability of \emph{individual} learned cubes which is independent from the derivation
DAG $G$. To this end, we need a condition which is sufficient to conclude that
some \emph{arbitrary} cube $C$ is derivable with respect to a PCNF $\psi$, i.e.~to check
whether $\psi \vdash C$. However, we are not aware of such a condition.

As an alternative to keeping the full derivation DAG in memory, a \emph{fresh}
selector variable can be added to \emph{each} newly learned initial
cube. Similar to selector variables in clauses, these variables are
transferred to all derived cubes.  Potentially non-derivable  cubes are
then disabled by assigning the selector variables accordingly.  
However, different from clauses, it must
be checked \emph{explicitly} which initial cubes are non-derivable by checking
the condition in Definition~\ref{def_model} for all initial cubes in the set
$\LCU_i$ of learned cubes. This amounts to an asymmetric treatment of selector
variables in clauses and cubes. Clauses are added to and removed from the CNF
part by \emph{push} and \emph{pop} operations provided by the solver API. This
way, it is known precisely which clauses are removed. In contrast to that,
cubes are added to the set of learned cubes $\LCU_i$ on the fly during cube
learning. Moreover, the optimization based on the temporal ordering of
selector variables from the previous section is not applicable to generate
shorter cubes since cubes are not associated to stack frames. 

Due to the complications illustrated above, we implemented the following
simple approach in \depqbf to keep only initial 
cubes. Every initial cube computed by the solver is stored in
a linked list $L$ of bounded capacity, which is increased
dynamically. The list $L$ is separate from the set of learned
clauses. 
Assume that a set $\phi^{\mathit{add}}_{i+1}$ of clauses is added to
the CNF part $\phi_i$ of the current PCNF to obtain the CNF part
$\phi_{i+1} = (\phi_{i} \setminus \phi^{\mathit{del}}_{i+1} ) \cup
\phi^{\mathit{add}}_{i+1}$ of the next PCNF $\psi_{i+1} =
\hat{Q}_{i+1}.\,\phi_{i+1}$. All the cubes in the current set $\LCU_i$
of learned cubes are discarded. For every added clause $C \in
\phi^{\mathit{add}}_{i+1}$ and for every  initial cube $C'
\in L$, it is checked whether the assignment $A$ given by 
$C'$ is a model
of the next PCNF $\psi_{i+1}$.  Initial cubes $C'$
for which this check succeeds are added to the set $\LCU_{i+1}$ of
learned cubes in the ACNF $\psi_{i+1}'$ of the next PCNF
$\psi_{i+1}$ after existential reduction has been applied to them. If the
check fails, then $C'$ is removed from $L$. It suffices to check the
initial cubes in $L$ only
with respect to the clauses $C \in
\phi^{\mathit{add}}_{i+1}$, and not the full CNF part $\phi_{i+1}$,
since the assignments given by the cubes in $L$ are models of the
\emph{current} PCNF $\psi_{i}$. In the end,
the set $\LCU_{i+1}$ contains only initial cubes all of which are derivable with respect to the ACNF
$\psi_{i+1}'$. If clauses are removed from the formula, then by
Proposition~\ref{prop_cleaned_up_cubes_sound} variables which
do not occur anymore in the formula are removed from the initial
cubes in $L$.

In the incremental QBF-based approach to BMC for partial 
designs~\cite{DBLP:conf/sat/MarinMB12,DBLP:conf/date/MarinMLB12}, all cubes
are kept across different solver calls under the restriction that the
quantifier prefix is modified only at the left end. This
restriction does not apply to incremental solving of PCNF where the formula can
be modified arbitrarily.


\subsection{Incremental QBF Solver API} \label{sec_inc_api}

The API of \depqbf~\cite{LonsingEglyICMS14} provides
functions to manipulate the prefix and the CNF part of the current
PCNF. Clauses are added and removed by the \emph{push} and \emph{pop}
operations described in Section~\ref{sec_cnf_stack}.  New quantifier blocks
can be added at any position in the quantifier prefix. New variables can be
added to any quantifier block. Variables which no longer occur in the formula
and empty quantifier blocks can be explicitly deleted. The quantifier block
$B_0$ containing the frame selector variables is invisible to the user.  
The solver maintains the learned constraints as described in Sections~\ref{sec_clause_del}
and~\ref{sec_clause_add} without any user interaction.

The \emph{push} and \emph{pop} operations are a feature of
\depqbf. Additionally, the API supports the manual insertion of selector
variables into the clauses by the user. Similar to incremental SAT
solving~\cite{DBLP:journals/entcs/EenS03}, clauses can then be enabled and
disabled manually by assigning the selector variables as assumptions via the
API. In this case, these variables are part of the QBF encoding and the
optimization based on the frame ordering presented in
Section~\ref{sec_clause_del} is not applicable. After a PCNF has
been found unsatisfiable (satisfiable) under assumptions where
the leftmost quantifier block is existential (universal),
the set of relevant assumptions which were used by the solver
to determine the result can be extracted.\footnote{This is similar to the function ``analyzeFinal'' in
\minisat, for example.}


\section{Experimental Results}

\setlength{\tabcolsep}{3pt}

\begin{table}[t]
\begin{minipage}{0.5\textwidth} 
\centering
\begin{tabular}{|lccc|}
\hline
\multicolumn{4}{|c|}{\textbf{QBFEVAL'12-SR}} \\
 & \emph{discard LC} & \emph{keep LC} & \emph{diff.(\%)} \\ 
$\overline{a}$: & $29.37 \times 10^6$  & $26.18 \times 10^6$  & -10.88 \\ 
$\tilde{a}$: & 3,833,077  & 2,819,492  & -26.44 \\ 
\hline
$\overline{b}$: & 139,036  & 116,792 & -16.00 \\ 
$\tilde{b}$: & 8,243  & 6,360  & -22.84 \\ 
\hline
$\overline{t}$: & 99.03  & 90.90 & -8.19 \\ 
$\tilde{t}$: & 28.56 & 15.74 & -44.88 \\ 
\hline
\end{tabular}
\end{minipage}
\begin{minipage}{0.5\textwidth}
\centering
\begin{tabular}{|lccc|}
\hline
\multicolumn{4}{|c|}{\textbf{QBFEVAL'12-SR-Bloqqer}} \\
 & \emph{discard LC} & \emph{keep LC} & \emph{diff.(\%)} \\ 
$\overline{a}$: & $39.75 \times 10^6$  & $34.03 \times 10^6$  & -14.40 \\ 
$\tilde{a}$: & $1.71 \times 10^6$  &  $1.65 \times 10^6$ & -3.62 \\ 
\hline
$\overline{b}$: & 117,019 & 91,737 & -21.61 \\ 
$\tilde{b}$: & 10,322 & 8,959 & -13.19 \\ 
\hline
$\overline{t}$: & 100.15 & 95.36 & -4.64 \\ 
$\tilde{t}$: & 4.18 & 2.83 & -32.29 \\ 
\hline
\end{tabular}
\end{minipage}
\caption{Average and median number of assignments ($\overline{a}$ and
$\tilde{a}$, respectively), backtracks ($\overline{b},\tilde{b}$), and wall
clock time ($\overline{t},\tilde{t}$) in seconds on \emph{sequences} $S = \psi_0, \ldots,
\psi_{10}$ of PCNFs which were fully solved by \depqbf both if all learned
constraints are discarded (\emph{discard LC}) and if constraints which are
correct in the sense of Propositions~\ref{prop_learning_satequiv} and~\ref{prop_cleaned_up_cubes_sound} are kept
(\emph{keep LC}). Clauses are \emph{added} to $\psi_i$ to obtain $\psi_{i+1}$ in $S$.}
\label{tab_experiments_eval2012}
\end{table}
\begin{table}[t]
\begin{minipage}{0.5\textwidth} 
\centering
\begin{tabular}{|lccc|}
\hline
\multicolumn{4}{|c|}{\textbf{QBFEVAL'12-SR}} \\
 & \emph{discard LC} & \emph{keep LC} & \emph{diff.(\%)} \\ 
$\overline{a}$: & $5.48 \times 10^6$  & $0.73 \times 10^6$  & -86.62 \\ 
$\tilde{a}$: & 186,237  &  15,031 & -91.92 \\ 
\hline
$\overline{b}$: & 36,826 & 1,228 & -96.67 \\ 
$\tilde{b}$: & 424 & 0 & -100.00 \\ 
\hline
$\overline{t}$: & 21.94 & 4.32 & -79.43 \\ 
$\tilde{t}$: & 0.75 & 0.43 & -42.66 \\ 
\hline
\end{tabular}
\end{minipage}
\begin{minipage}{0.5\textwidth}
\centering
\begin{tabular}{|lccc|}
\hline
\multicolumn{4}{|c|}{\textbf{QBFEVAL'12-SR-Bloqqer}} \\
 & \emph{discard LC} & \emph{keep LC} & \emph{diff.(\%)} \\ 
$\overline{a}$: & $5.88 \times 10^6$  & $1.29 \times 10^6$  & -77.94 \\ 
$\tilde{a}$: & 103,330  & 8,199  & -92.06 \\ 
\hline
$\overline{b}$: & 31,489 & 3,350 & -89.37 \\ 
$\tilde{b}$: & 827 & 5 & -99.39 \\ 
\hline
$\overline{t}$: & 30.29  & 9.78 & -67.40 \\ 
$\tilde{t}$: & 0.50 & 0.12 & -76.00 \\ 
\hline
\end{tabular}
\end{minipage}
\caption{Like Table~\ref{tab_experiments_eval2012} but for the reversed sequences
  $S' = \psi_{9}, \ldots, \psi_0$ of PCNFs after the original sequence $S =
  \psi_0, \ldots, \psi_9,\psi_{10}$ has been solved. Clauses are \emph{deleted} from $\psi_i$ to obtain $\psi_{i-1}$ in $S'$.}
\label{tab_experiments_eval2012_pop}
\end{table}

To demonstrate the basic feasibility of general incremental QBF solving, 
we evaluated our incremental QBF solver \depqbf based on the instances from
\emph{QBFEVAL'12 Second Round (SR)} with and without preprocessing by
\bloqqer.\footnote{\url{http://www.kr.tuwien.ac.at/events/qbfgallery2013/benchmarks/}.}
We disabled the sophisticated dependency analysis in terms of dependency
schemes in \depqbf and instead applied the linear ordering of the quantifier
prefix in the given PCNFs.  For experiments, we constructed a
sequence of related PCNFs for \emph{each} PCNF in the benchmark sets as
follows. Given a PCNF $\psi$, we divided the number of clauses in $\psi$ by 10
to obtain the size of a slice of clauses. The first PCNF $\psi_0$ in the
sequence contains the clauses of one slice. The clauses of that slice are
removed from $\psi$. The next PCNF $\psi_1$ is obtained from $\psi_0$ by
adding another slice of clauses, which is removed from $\psi$. The other PCNFs
in the sequence $S = \psi_0,\psi_1, \ldots, \psi_{10}$ are constructed similarly
so that finally the last PCNF $\psi_{10}$ contains all the clauses from the
original PCNF $\psi$. In our tests, we constructed each PCNF $\psi_i$ from the
previous one $\psi_{i-1}$ in the sequence by adding a slice of clauses to a
new frame after a \emph{push} operation. 
We ran \depqbf on the sequences of PCNFs constructed this way with a wall clock time limit of 1800 seconds and a
memory limit of 7 GB. 

Tables~\ref{tab_experiments_eval2012}
and~\ref{tab_experiments_eval2012_pop} show experimental
results\footnote{Experiments were run on AMD Opteron 6238, 2.6 GHz, 64-bit
Linux.} on sequences $S = \psi_0, \ldots, \psi_{10}$ of PCNFs and on the reversed
ones $S' = \psi_9, \ldots, \psi_0$, respectively. To generate $S'$, we first
solved the sequence $S$ and then started to discard
clauses by popping the frames from the clause stack of DepQBF via its API.  In one run 
(\emph{discard LC}), we always discarded all the constraints that were learned
from the previous PCNF $\psi_i$ so that the solver solves the next PCNF
$\psi_{i+1}$ ($\psi_{i-1}$ with respect to Table~\ref{tab_experiments_eval2012_pop})
starting with empty sets of learned clauses and cubes. In another run
(\emph{keep LC}), we kept learned constraints  
 as described in Sections~\ref{sec_clause_del} and \ref{sec_clause_add}. This way, 70 out of 345
total PCNF sequences were fully solved from the set \emph{QBFEVAL'12-SR} by
both runs, and 112 out of 276 total sequences were fully solved from the set
\emph{QBFEVAL'12-SR-Bloqqer}.

The numbers of assignments, backtracks, and wall clock time indicate that
keeping the learned constraints is beneficial in incremental QBF solving
despite the additional effort of checking the collected initial cubes.  In the
experiment reported in Table~\ref{tab_experiments_eval2012} clauses are always
added but never deleted to obtain the next PCNF in the sequence. Thereby,
across all incremental calls of the solver in the set \emph{QBFEVAL'12-SR} on
average 224 out of 364 (61\%) collected initial cubes were identified as
derivable and added as learned cubes. For the set \emph{QBFEVAL'12-SR-Bloqqer},
232 out of 1325 (17\%) were added.

Related to Table~\ref{tab_experiments_eval2012_pop}, clauses are always
removed but never added to obtain the next PCNF to be solved, which allows to
keep learned cubes based on Proposition~\ref{prop_cleaned_up_cubes_sound}. Across all incremental calls of the solver
in the set \emph{QBFEVAL'12-SR} on average 820 out of 1485 (55\%) learned
clauses were disabled and hence effectively discarded because their
Q-resolution derivation depended on removed clauses. For the set
\emph{QBFEVAL'12-SR-Bloqqer}, 704 out of 1399 (50\%) were disabled.


\section{Conclusion}

We presented a general approach to incremental QBF solving which integrates
ideas from incremental SAT solving and which can be implemented in any
QCDCL-based QBF solver. The API of our incremental QBF solver \depqbf provides
\emph{push} and \emph{pop} operations to add and remove clauses in a PCNF. This increases the usability of our
implementation. Our approach is application-independent and applicable to
arbitrary QBF encodings.

We illustrated the problem of keeping the learned constraints across different
calls of the solver. To improve cube learning in incremental QBF solving, it
might be beneficial to maintain (parts of) the cube derivation in memory. This
would allow to check the cubes more precisely than with the simple approach we
implemented. Moreover, the generation of proofs and certificates~\cite{DBLP:journals/fmsd/BalabanovJ12,DBLP:conf/ijcai/GoultiaevaGB11,DBLP:conf/sat/NiemetzPLSB12} is supported if the derivations are kept in memory rather than in a trace file.

Dual
reasoning~\cite{DBLP:conf/sat/GoultiaevaB13,DBLP:conf/date/GoultiaevaSB13,DBLP:conf/sat/KlieberSGC10,DBLP:conf/cp/Gelder13}
and the combination of preprocessing and certificate
extraction~\cite{DBLP:conf/lpar/JanotaGM13,DBLP:conf/sat/MarinMB12,DBLP:conf/date/SeidlK14}
are crucial for the performance and applicability of CNF-based QBF
solving. The combination of incremental solving with these techniques has the
potential to further advance the state of QBF solving.

Our experimental analysis demonstrates the feasibility of incremental QBF
solving in a general setting and motivates further applications, along with the study of BMC of partial
designs using incremental QBF solving~\cite{DBLP:conf/date/MarinMLB12}.   
Related experiments with conformant planning based on incremental solving 
by \depqbf showed promising results~\cite{DBLP:journals/corr/EglyKLP14}. 
Further experiments
with problems which are inherently incremental can provide more insights and
open new research directions.



\end{document}